\documentclass[nofootinbib,11pt,notitlepage,superscriptaddress]{revtex4-1}
\usepackage{amsmath,amssymb,amsfonts}
\usepackage{amsthm}
\usepackage{graphicx}
\usepackage{color}

\usepackage{enumerate}
\newcommand{{\one}}{{\mathbf{1}}}

\newtheorem{theorem}{Theorem}
\newtheorem{lemma}{Lemma}
\newtheorem{corollary}{Corollary}
\theoremstyle{remark}
\newtheorem*{remark}{Remark}
\newtheorem{definition}{Definition}
\begin{document}

\begin{abstract}

In R.D. Sorkin's framework for logic in physics a clear
separation is made between the collection of 
unasserted propositions about the physical world and the affirmation or denial of these
propositions by the physical world. The unasserted propositions 
form a Boolean algebra because they correspond to subsets of an underlying set of 
spacetime histories. \textit{Physical} rules of inference, apply not to the 
propositions in themselves but to the affirmation and denial 
of these propositions by the actual world. This \textit{physical logic}
may or may not respect  the propositions' underlying Boolean 
structure. We prove that this logic is Boolean if and 
only if the following three axioms hold: (i) The world is affirmed, (ii) Modus Ponens and (iii) If a
proposition
is denied then its negation, or complement, is affirmed. 
When a physical system is governed by a dynamical law in the form of a  quantum measure 
with the rule that events of zero measure are denied, the 
axioms (i) - (iii) prove to be too rigid and need to be modified. One promising scheme for 
quantum mechanics as quantum measure theory corresponds to replacing axiom (iii) 
with axiom (iv) Nature is as fine grained as the dynamics allows.

\end{abstract}

\title{Physical Logic}
\author{Kate Clements}
\affiliation{Blackett Laboratory, Imperial College, London SW7 2AZ, U.K.}
\author{Fay Dowker}
\affiliation{Blackett Laboratory, Imperial College, London SW7 2AZ, U.K.}
\affiliation{Perimeter Institute, 39 Caroline St.\ N., Waterloo, ON N2L 2Y5, Canada}
\affiliation{Institute for Quantum Computing, University of Waterloo, ON N2l 3G1, Canada}
\author{Petros Wallden}
\affiliation{School of Informatics, University of Edinburgh, Edinburgh EH8 9AB, UK}
\affiliation{IPaQS, Heriot-Watt University, Edinburgh EH14 1AS, UK}

\maketitle

\section{Introduction}
The view that the mode of reasoning we use for classical physics is not appropriate when discussing a quantum system is widespread, if not mainstream.
For example, in the Quantum Mechanics volume of his \emph{Lectures on Physics}, R.P.~Feynman refers to ``the logical tightrope on which we must walk if we wish to describe nature successfully'' \cite{FeynmanLecs}.
In order to investigate the nature of this ``tightrope'' further, in a systematic way, we need a framework for logic that is relevant for \emph{physics} (rather than, say, mathematics or language) and within which the logic used for classical physics can be identified, characterised, assessed and, if necessary, replaced.
Recently, a unifying foundation for physical theories with spacetime character --- Generalised Measure Theory (GMT) --- which provides just such a framework has been set out 
\cite{Sorkin94,Sorkin97,Sorkin07a,Sorkin07b}.
The key to the clarity that this formalism brings to the study of deductive inference in physics is the distinction it makes between the \emph{assertion} of propositions about the physical world and the propositions themselves, the latter corresponding merely to questions waiting to be answered \cite{Sorkin10,Sorkin11}.
Identifying the answers to the questions as the physical content of the theory, as explained below, makes it a small step to consider the possibility of non-standard rules of inference; to do so is to open a new window on the variety of antinomies with which quantum mechanics is so infamously plagued (or blessed) \cite{Sorkin07a,Dowker08,Dowker11,Henson11,MikeDion}.

We begin in Section \ref{sec:basics} by  
identifying three basic structures that constitute a general framework for reasoning about the physical world. 
Taking the revolution of relativity seriously, we
 assume that the physical theory has a \textit{spacetime} character in the sense that it is based on a set of spacetime histories  which represent the finest grained descriptions of the system conceivable within the theory. 
In Section \ref{sec:rules} we give names to certain rules of inference and situate classical, Boolean rules of inference within this framework.
In Section \ref{sec:tlp} we investigate which rules are implied by which others, abstractly, by mathematical manipulation alone, setting aside the question of which might be necessary or desirable for 
physics. We
focus on the rule of inference known as \emph{modus ponendo ponens} (\emph{modus ponens} for short), the basis for deductive proof without which the ability to reason at all might seem to be compromised from the outset\footnote{Lewis Carroll gives in \cite{Carroll1895} a witty account of the implications of a failure to take up modus ponens explicitly as a rule of inference.}. We will show that -- on the 
mildest conceivable assumption that something happens in the world -- modus ponens
implies Boolean logic \textit{if} it is supplemented by the rule ``If a proposition is denied by the
physical world, then its negation is affirmed.'' 
These results are independent of whether the theory is classical, quantum or transquantum in Sorkin's hierarchy of 
physical theories \cite{Sorkin94}.
We will show that one currently favoured scheme for
interpreting quantum theory, the \emph{multiplicative scheme}, coincides with the adoption of modus ponens, together with a condition of finest grainedness.

\section{The three-fold structure}\label{sec:basics}
The details of any logical scheme for physics---in particular, the events about which it is intended to reason---will plainly depend on the system one has in mind.
Nonetheless, one can describe a class of schemes rather generally in terms of three components \cite{Sorkin07b}:
\begin{enumerate}[(i)]
\item the set, $\mathfrak{A}$, of all  \emph{questions} that can be asked about the system;
\item the set, $S$, of possible \emph{answers} to those questions; and
\item a collection, $\mathfrak{A^{*}}$, of \emph{answering maps} $\phi: \mathfrak{A} \rightarrow S $, exactly one of which corresponds to the physical world.
\end{enumerate}
While such a framework may not be the most general that could be conceived, it is broad enough to encompass all classical theories, including stochastic theories such as Brownian motion.
In such a classical physical theory $\mathfrak{A}$ is a Boolean algebra, $S = \mathbb{Z}_2 = \{0,1\}$ and $\phi$ is a homomorphism from $\mathfrak{A}$ into $\mathbb{Z}_{2}$, as we describe below.
To make the framework general enough to include quantum theories one might consider altering any or all of these three classical ingredients.
It is remarkable that the only change that appears to be necessary in order to accommodate quantum theories in a spacetime approach based on the Dirac--Feynman path integral is to free $\phi$ from the 
constraint that it be a homomorphism \cite{Sorkin07a, Sorkin07b, Sorkin11}.
We will assume the following about the three components.
\begin{enumerate}[(i)] 

\item  $\mathfrak{A}$ is a Boolean algebra, which we refer to as the \emph{event algebra}.
The elements of $\mathfrak{A}$ are equivalently and interchangeably referred to as \emph{propositions} or \emph{events}, where it is understood that these terms refer to \emph{unasserted} propositions. 
An  event $A$ can also be thought of as corresponding to the \textit{question}, ``Does $A$ happen?''.
The elements of $\mathfrak{A}$ are subsets of the set, $\Omega$, of spacetime histories of the physical system. For instance, in Brownian motion $\Omega$ is the set of Wiener paths.
Use of the term `event' to refer to a subset of $\Omega$ is standard for stochastic processes.
In the quantal case, $\Omega$ is the set of histories summed over in the path integral, for example the particle trajectories in non-relativistic quantum mechanics.
An example of an event in that case is the set of all trajectories in which the particle passes through some specified region of spacetime.

The Boolean operations of meet $\wedge$ and join $\vee$ are identified with the set operations of intersection $\cap$ and union $\cup$, respectively.
A note of warning: using in this context the words \emph{and} and \emph{or} to denote the algebra elements that result from these set operations can lead to ambiguity.
In this paper we will try to eliminate the ambiguity by the use of single inverted commas, so that `$A$ or $B$' denotes the event  $A\vee B$; `$A$ and $B$', the event $A\wedge B$.

The zero element $\emptyset\in \mathfrak{A}$ is the empty set and the unit element $\one\in\mathfrak{A}$ is $\Omega$ itself.
The operations of multiplication and addition of algebra elements are, respectively,
\begin{align*}
AB&:=A\cap B, \ \ \forall A, B \in \mathfrak{A};\\
A+B&:=(A\setminus B)\cup(B\setminus A), \ \ \forall A, B \in \mathfrak{A}\,.
\end{align*}
With these operations, $\mathfrak{A}$ \emph{is} an algebra in the sense of being a vector space over the finite field $\mathbb{Z}_2$.
A useful expression of the subset property is: $A \subseteq B \Leftrightarrow AB = A$.
We have, for all $A$ in $\mathfrak{A}$,
\begin{align}
AA&=A;\\
A+A&=\,\emptyset;\:\text{and}\\
{\neg}A&:=\Omega \setminus A = \one+A\,.
\end{align}
The event $\one+A$ may be referred to as $\neg A$, as the \emph{complement} of $A$, or again with single inverted commas, as `not $A$'\footnote{See \cite{Sorkin10} for a discussion of the ambiguity in the phrase ``not $A$''.}.
\item Together with the algebra of questions comes the space of potential \emph{answers} that the physical system can provide to those questions.
Whilst one can envisage any number of generalisations, with intermediate truth values for example, we follow Sorkin and keep as the answer space that of classical logic, namely the Boolean algebra $\mathbb{Z}_{2}\equiv\lbrace 0,1\rbrace\equiv\lbrace\text{false, true}\rbrace\equiv\lbrace\text{no, yes}\rbrace$.
To answer the question $A \in \mathfrak{A}$ with 1 (0) is to assert that the event $A$ does (does not) happen; equivalently, we say that $A$ is affirmed (denied).
\item Finally, one has the set $\mathfrak{A}^{*}$ of allowed \emph{answering maps}, also called \emph{co-events}\footnote{The notation $\mathfrak{A}^{*}$ reflects the nature of the co-event space as dual to the event algebra.}, $\phi:\mathfrak{A}\to\mathbb{Z}_{2}$.
We assume that a co-event is a non-constant map: $\phi \ne 0$ and $\phi \ne 1$.
That is, a co-event must affirm at least one event and deny at least one event.
To specify a co-event is to answer every physical question about the system, and thus to give as complete an account of what happens as one's theory permits.
The physical world corresponds to exactly one co-event from $\mathfrak{A}^{*}$. In other words, the physical world provides (or is equivalent to) a definite answer to every question and $\mathfrak{A}^{*}$ is the set of possible physical worlds. 
\end{enumerate}
This three-fold structure of event algebra $\mathfrak{A}$, answer space $\mathbb{Z}_2$ and collection of answering maps or co-events $\phi: \mathfrak{A} \rightarrow \mathbb{Z}_2$ makes sense out of possibilities that seem otherwise non-sensical \cite{Sorkin10}.
The three-fold structure is appropriate to physics, where perfectly sensible, meaningful events are not in themselves true or false (unlike, for example and on one view, mathematical statements).
Each event will either happen or not \emph{in the physical world}, but which it is is contingent.

This seems an appropriate point to note that the three-fold structure is also apparent in the framework for interpreting quantum physics
that came commonly to be known 
as ``Quantum Logic." There are major differences with
Sorkin's framework however.
In place of the Boolean algebra of propositions there is in Quantum Logic
an orthocomplemented lattice of subspaces of a Hilbert space. In place of the set of yes/no answers
 $\mathbb{Z}_2$,
there is the set of probabilities, real numbers between 0 and 1. And in place of the co-event, there 
is the state which maps each subspace of Hilbert space to a probability. Thus, at the very outset, the 
space of propositions in Quantum Logic has a non-Boolean character due to the focus on Hilbert space as the 
arena for the physics. Quantum Logic sprang from a canonical approach to quantum theory 
in which Hilbert space is fundamental. Hilbert space has 
no place in classical physics, hence the starting point -- the set of
unasserted propositions -- for Quantum Logic is different than in classical physics. In contrast, 
in a path integral approach to quantum physics as adopted by Sorkin, the axiomatic basis is a set of spacetime histories just as it is in classical physics and therefore the structure of the 
set of unasserted propositions is the same in both classical and quantum physics:
this is a \textit{unifying} framework.

\section{Rules of inference} \label{sec:rules}
If we somehow came to know the co-event that corresponds to the whole universe, then there would be no need for rules of inference: we would know everything already.
Rules of inference are needed because our knowledge is partial and limited and to extend that knowledge further we need to be able to deduce new facts from established ones.
As stressed by Sorkin, on this view dynamical laws in physics are rules of inference \cite{Sorkin11}: using the laws of gravity, we can infer from the position of the moon tonight its position yesterday and its position tomorrow.

For the purposes of this paper we call any condition restricting the collection of allowed co-events a rule of inference\footnote{An alternative is to call any condition on the allowed co-events a dynamical law.}.
One could begin by considering the set of \emph{all} non-constant maps from $\mathfrak{A}$ into $\mathbb{Z}_2$; a rule of inference is then any axiom that reduces this set.
One axiom that has been suggested \cite{Sorkin07a, Sorkin07b} is that of \emph{preclusion}, the axiom that an event of zero measure is denied.
Explicitly, if $\mu$ is the (classical, quantum or transquantum) measure on the event algebra, encoding both the dynamics and the initial conditions for the system, then $\mu(A) = 0 \Rightarrow \phi(A) = 0$.
We will return to preclusion in a later Section; for the time being, our attention will focus on rather more structural axioms.

First, let us define some properties of co-events that it might be desirable to impose.
In all the following definitions, $\phi$ is a co-event and we recall that $\phi$ is assumed not to be the zero map or unit map.
We begin with properties that reflect the algebraic structure of $\mathfrak{A}$ itself.
\begin{definition}\label{def:zero}
$\phi$ is \emph{zero-preserving} if
\begin{equation} \phi(\emptyset)=0 \label{eq:zero}\,.
\end{equation}
\end{definition}
\begin{definition}\label{def:unital}
$\phi$ is \emph{unital} if
\begin{equation}
\phi({\one})=1 \label{eq:unital}\,.
\end{equation}
One could call this condition ``the world is affirmed.''
\end{definition}
\begin{definition}\label{def:mult}
$\phi$ is \emph{multiplicative} if
\begin{equation}
\phi(AB)=\phi(A)\phi(B), \ \ \ \forall A, B \in \mathfrak{A} \label{eq:mult}\,.
\end{equation}
\end{definition}
\begin{definition}\label{def:add}
$\phi$ is \emph{additive} or \emph{linear} if
\begin{equation}
\phi(A+ B)=\phi(A)+ \phi(B), \ \ \ \forall A, B \in \mathfrak{A} \label{eq:add}\,.
\end{equation}
\end{definition}

A further set of conditions is motivated directly as the formalisation of the rules of inference that we use in  classical reasoning.
As mentioned in the Introduction, arguably the most desirable among these is modus ponens, commonly stated thus:
\begin{quote}
MP: \emph{If $A$ implies $B$ and $A$,  then $B$.}
\end{quote}
However, it is now easy to appreciate why care must be taken in distinguishing (mere unasserted) events from statements about the physical world, {\textit{i.e.}} affirmed or denied events.
The rules of inference we are interested in here are those that tell us how to deduce statements about what happens in the physical world from other such statements.
To render modus ponens fully in terms of the three-fold framework for physics, we re-express it as
\begin{quote}
MP: \emph{If `$A$ implies $B$' is affirmed and $A$ is affirmed, then $B$ is affirmed},
\end{quote}
where  `$A$ implies $B$' is an event, an element of $\mathfrak{A}$ which we denote symbolically as $A\to B := \neg (A \wedge(\neg B))$. We have,
\begin{align}
A\to B &= \one+A(\one+B)\nonumber\\
&=\one+A+AB,
\end{align}
which small manipulation shows, incidentally, how much easier it is to work with the arithmetic form of the operations than with $\wedge$, $\vee$ and $\neg$.
The condition of modus ponens is then:
\begin{definition}\label{def:MP}
 $\phi$ is MP if
\begin{equation}
\phi(A\to B)=1\,,\ \phi(A)=1\;\Rightarrow\;\phi(B)=1,\ \ \ \forall A, B \in \mathfrak{A}\label{eq:MP}\,.
\end{equation}
\end{definition}

Distinct from MP and from each other are the two strains of ``proof by contradiction,'' which we shall call C1 and C2.
In words, we can state them as follows:
\begin{quote}
C1: \emph{If event $A$ is affirmed, then its complement is denied.}
\end{quote}
\begin{quote}
C2: \emph{If event $A$ is denied, then its complement is affirmed}

\end{quote}
We point out that C1 and C2 are referred to as the \emph{law of contradiction} and the \emph{law of the excluded middle}, respectively, in \cite{RescherBrandom}. While closely related, these two notions are distinct and play a very different role in our analysis, with C2 being more important. As we will see in the next sections, C2 plays a complementary role to other conditions (MP, multiplicativity), while C1 is implied from those conditions\footnote{Note also, that at the level of the Boolean algebra of events we always have $\neg\neg A = A$ and, moreover, if we take ``law of the excluded middle'' to mean that every event
is either affirmed or denied then our three-fold framework respects it by fiat because that is just 
the statement that $\phi$ is a map to $\mathbb{Z}_2$ \cite{Sorkin10}. This illustrates how careful one must be to be clear.}. 
Formally:
\begin{definition}\label{def:c1}
$\phi$ is C1 if
\begin{equation}
\phi(A)=1\;\Rightarrow\; \phi(\one+A)=0,\ \ \ \forall A \in \mathfrak{A} \label{eq:C1}.
\end{equation}
\end{definition}
\begin{definition}\label{def:c2}
$\phi$ is C2 if
\begin{equation}
\phi(A)=0\;\Rightarrow\; \phi(\one+A)=1,\ \ \ \forall A \in \mathfrak{A} \label{eq:C2}.
\end{equation}
\end{definition}

\subsection{An example: classical physics, classical logic}

In classical physics we use classical logic because, in classical physics, One History Happens.
Indeed, the rules of inference known collectively as classical, Boolean logic follow from the axiom that the physical world corresponds to exactly one history in $\Omega$ \cite{Sorkin07b}.
In a later Section we will give a list of equivalent forms of this axiom in the case of finite $\Omega$; here, we note only that if the physical world corresponds to history $\gamma\in \Omega$ then all physical questions can be answered.
In other words, $\gamma$ gives rise to a co-event $\gamma^*: \mathfrak{A}\rightarrow \mathbb{Z}_2$, as
\begin{align}
\gamma^{*}(A)=&\begin{cases}
1&\text{if $\gamma \in A$}\\
0&\text{otherwise,}
\end{cases}\\
&\forall A\in\mathfrak{A}\,.\nonumber
\end{align}
It can be shown that such a $\gamma^*$ is both multiplicative and additive, \emph{i.e.} it is a homomorphism from $\mathfrak{A}$ into $\mathbb{Z}_2$.
It is easy to see that $\gamma^*$ is zero-preserving and unital, and one can further use its homomorphicity to prove it is C1, C2 and MP.
For example, we have
\begin{lemma}\label{lemma:c1c2}
If co-event $\phi$ is additive and unital then it is zero-preserving, C1 and C2.
\end{lemma}
\begin{proof}
Let $\phi$ be additive and unital. Then
\begin{equation*}
1 = \phi({\one}) = \phi({\one} + A + A) = \phi({\one}+ A) + \phi(A),\ \  \forall A\in\mathfrak{A}\,,
\end{equation*}
which implies that exactly one of $\phi({\one}+ A)$ and $\phi(A)$ is equal to 1.
So $\phi$ is C1 and C2, and $\phi(\emptyset) = 0$.
\end{proof}

If One History Happens, as in classical physics, the physical world fully respects the Boolean structure of the event algebra, and the logical connectives, \emph{and}, \emph{or}, \emph{not} and so forth may be used carelessly, without the need to specify whether they refer to asserted or to unasserted propositions.
One doesn't have to mind one's logical Ps and Qs over (potentially ambiguous) statements such as ``$A$ or $B$ happens'' when $\phi$ is a homomorphism:
\begin{lemma}\label{lemma:AorB}
If co-event $\phi$ is a  homomorphism then
$\phi(A\vee B) = 1 \iff \phi(A) = 1$ or $ \phi(B) =1$.
\end{lemma}
\begin{proof}
\begin{align*}
& \phi(A\vee B) = 1 \\
\iff &\phi(AB + A+B) = 1 \\
\iff &\phi(A)\phi(B) + \phi(A)+\phi(B) = 1\\
\iff &(\phi(A) + 1)(\phi(B) +1) = 0 \,.\qedhere
\end{align*}
\end{proof}

So no ambiguity arises because `$A$ or $B$' happens if and only if  $A$ happens or $B$ happens.

\section{Results}\label{sec:tlp}
\begin{theorem}\label{theorem:MPfiltermult}
The following conditions on a co-event $\phi$ are equivalent:
\begin{enumerate}[(i)]
\item $\phi$ is MP and unital;
\item $\phi^{-1}(1):=\lbrace A\in\mathfrak{A}\:|\:\phi(A)=1\rbrace$ is a filter\footnote{We assume a filter is non-empty and not equal to the whole of $\mathfrak{A}$.};
\item $\phi$ is multiplicative.
\end{enumerate}
\end{theorem}
\begin{proof}
\begin{enumerate}[]
\item (i) $\Rightarrow$ (ii)

Let $\phi$ be MP and unital.

First, we show that the superset of an affirmed event is affirmed.
Let $\phi(A) =1$ and $B$ be such that $AB = A$.
Then
\begin{align*}
\phi(A\to B)&=\phi(1+A+A)\nonumber\\
&=\phi(1)\nonumber\\
&=1,
\end{align*}
by unitality.
By MP it follows that $\phi(B)=1$.

Now we show that the intersection of two affirmed events is also affirmed.
Let $\phi(C)=\phi(D) =1$.
We have that $D(C\to CD) = D(1+C+CD) = D$, and so $\phi(D(C\to CD)) =1$.
By the first part of the proof, this implies that $\phi(C\to CD) =1$, so that by MP $\phi(CD) =1$.

Finally, $\phi$ is unital so $\phi^{-1}(1)$ is non-empty and $\phi\ne 1$ so $\phi^{-1}(1)$ is not equal to $\mathfrak{A}$.

\item (ii) $\Rightarrow$ (iii)

Let $\phi^{-1}(1)$ be a filter and $A,B\in \mathfrak{A}$.
Then there are two cases to check.
\begin{enumerate}[(a)]
\item If $\phi(A)=\phi(B)=1$ then the filter property implies that $\phi(AB)=1$.
\item Assume without loss of generality that $\phi(A) = 0$.
Since $A$ is a superset of $AB$, we must therefore have that $\phi(AB) = 0$; otherwise, the filter property would lead to the conclusion that $\phi(A) =1$: a contradiction.
 \end{enumerate}
So $\phi$ is multiplicative.

\item (iii) $\Rightarrow$ (i)

Let $\phi$ be multiplicative.
Since $\phi\ne 0$, $\exists X\in \mathfrak{A}$ s.t. $\phi(X) = 1$.
Then,
\begin{equation*}
1 = \phi(X) = \phi(\one\cdot X) = \phi(\one)\phi(X) = \phi(\one),
\end{equation*}
so $\phi$ is unital.

Now suppose $\phi(A) = \phi(A\to B)=1$.
We have that $A(A\to B) = A(1 + A + AB) = AB$, and thus $\phi(AB)=\phi(A)\phi(A\to B)=1$.
It follows that $\phi(A)\phi(B) = 1$, so that $\phi(B)=1$.
So $\phi$ is MP.\qedhere
\end{enumerate}
\end{proof}

Note, however, the following.
\begin{remark}
MP alone is not enough to guarantee multiplicativity, as shown by the following example.
Consider the event algebra $\mathfrak{A}=\lbrace \emptyset, \one\rbrace$, and the co-event
\begin{align*}
\phi(\emptyset)&=1;\\
\phi(\one)&=0.
\end{align*}
MP is trivially satisfied: $\phi(\emptyset\to \one)=\phi(\one+\emptyset+\emptyset)=\phi(\one) =0$, while $\phi(\one\to\emptyset)=\phi(\one+\one+\emptyset)=\phi(\emptyset)=1$, but $\phi(\one)=0$.
So there \emph{is} no pair of events $A$ and $B$ such that $\phi(A\to B),\:\phi(A)=1$, \emph{i.e.}~for which we even need to check whether $\phi(B)=1$.
Multiplicativity fails, however:
\begin{align*}
\phi(\emptyset\cdot\one)&=\phi(\emptyset)=1\\
&\neq\phi(\emptyset)\phi(\one)=1\cdot0=0.
\end{align*}
Neither does MP together with \emph{zero-preservation} guarantee multiplicativity, as demonstrated again by an example.
Consider this time the four-element event algebra $\mathfrak{A}=\lbrace\emptyset,A,B,\one\rbrace$, where  $B=\one+A$, and the following zero-preserving co-event:
\begin{align*}
\phi(\emptyset)&=\phi(B)=\phi(\one)=0;\\
\phi(A)&=1.
\end{align*}
MP is trivially satisfied by an argument similar to that above, but $\phi$ is not multiplicative, since
\begin{align*}
\phi(A\cdot\one)&=\phi(A)=1\\
&\neq\phi(A)\phi(\one)=1\cdot0=0.
\end{align*}
\end{remark}

Having established the relation between multiplicativity of a co-event and the pillar of classical inference---MP---what can be said of the proofs by contradiction?
From the proof of Theorem \ref{theorem:MPfiltermult} we know that a multiplicative $\phi$ is unital.
It is also C1:
\begin{lemma}
\label{lemma:C1}
If $\phi$ is a multiplicative co-event then $\phi$ is zero-preserving and C1.
\end{lemma}
\begin{proof}
Let $\phi$ be multiplicative. $\phi\ne 1$, so $\exists A\in \mathfrak{A}$ s.t. $\phi(A) = 0$.
Thus
\begin{equation*}
\phi(\emptyset) = \phi\left(A(\one + A)\right) = \phi(A)\phi(\one + A) = 0\,.
\end{equation*}
Now let $\phi(B)=1$ for some $B\in\mathfrak{A}$.
Then
\begin{align*}
0&=\phi(\emptyset) = \phi\left(B(\one+B)\right) =\phi(B)\phi(\one+B)\\
&\Rightarrow\;\phi(\one + B)=0\,.\qedhere
\end{align*}
\end{proof}

\begin{corollary}
If the co-event $\phi$ is MP and unital  then it is C1.
\end{corollary}

It was shown in the previous Section that if  a co-event $\phi$ is a homomorphism then it is MP, C1 and C2.
Conversely, we can ask: what conditions imply that $\phi$ is a homomorphism?

\begin{theorem}\label{theorem:hom}
If co-event $\phi$ is unital, MP and C2 then it is a homomorphism.
\end{theorem}
\begin{proof}
Let $\phi$ be unital, MP and C2. By Theorem \ref{theorem:MPfiltermult} $\phi$ is multiplicative, and by Lemma \ref{lemma:C1} it is C1.

We need  to show that $\phi$ is additive.
C1 and C2 imply $\phi(X) + \phi({\one} + X) = 1$ for all $X\in \mathfrak{A}$.
Let $A,B\in \mathfrak{A}$.
\begin{align*}
{}& \phi(A+B) + \phi(\one+A+B) =1 \ \ \ {\textrm{and}}\ \ \ \phi(AB) + \phi(\one+AB) =1\\
\Rightarrow\, &\left[\phi(A+B) + \phi(\one+A+B)\right]\left[\phi(AB) + \phi(\one +AB)\right] =1\\
\Rightarrow\, & \phi(A+B)\phi(AB) + \phi(A+B)\phi(\one+AB) + \phi(\one+A+B)\phi(AB) + \phi(\one+A+B)\phi(\one +AB) =1 \\
\Rightarrow\, & \phi\left((A+B)AB\right) + \phi\left((A+B)(\one+AB)\right) + \phi\left((\one+A+B)AB\right) + \phi\left((\one+A+B)(\one +AB)\right) =1 \\
\Rightarrow\, & \phi(\emptyset) + \phi(A+B) + \phi(AB) + \phi(\one+A+B+AB) = 1\\
\Rightarrow\, & 0+ \phi(A+B) + \phi(AB) + \phi\left((\one+A)(\one+B)\right) = 1\\
\Rightarrow\, & \phi(A+B) + \phi(AB) + \phi(\one+A) \phi(\one +B) = 1\\
\Rightarrow\, & \phi(A+B) + \phi(AB) + (1 + \phi(A))(1+\phi(B)) = 1\\
\Rightarrow\, & \phi(A+B) + \phi(AB) + 1 + \phi(A) + \phi(B) + \phi(AB) = 1 \\
\Rightarrow\, & \phi(A+B) = \phi(A) + \phi(B).
\qedhere
\end{align*}
\end{proof}

Since zero-preservation and C2 imply unitality we can replace the condition of unitality by that of zero-preservation:
\begin{corollary}
If co-event $\phi$ is  zero-preserving, MP and C2 then it is a homomorphism.
\end{corollary}

Theorem \ref{theorem:hom} establishes that, as long as $\phi(\one) = 1$ (the world is affirmed), modus ponens needs the addition of only the rule C2 to lead to classical logic.

\section{A Unifying Proposal}
\subsection{Classical physics revisited}
We mentioned that when One History Happens, the corresponding co-event is a homomorphism.
What about the converse?
When the set of spacetime histories $\Omega$ is finite, the event algebra $\mathfrak{A}$ is the power set $2^\Omega$ of $\Omega$, and in this case the Stone representation theorem tells us that the set of (non-zero) homomorphisms from $\mathfrak{A}$ to $\mathbb{Z}_2$ is isomorphic to $\Omega$.
Thus, the axiom that exactly one history from $\Omega$ corresponds to the physical world is equivalent---in the finite case---to the assumption that the co-event that corresponds to the physical world is a homomorphism.
This is just one of the possible equivalent reformulations of the One History Happens axiom that defines classical physics; we provide a partial list below.
Before doing so we must first introduce classical \emph{dynamics} as a rule of inference.
The dynamics are encoded in a probability measure $\mu$, a non-negative real function $\mu: \mathfrak{A} \rightarrow \mathbb{R}$ which satisfies the Kolmogorov sum rules and $\mu(\mathbf{1}) = 1$.
We call an event in $\mathfrak{A}$ such that $\mu(A) =0$ a \emph{null} event.
Classical dynamical law requires that the history that corresponds to the physical world not be an element of any null event: a null event cannot happen.
The co-event $\phi$ that corresponds to the physical world is therefore required to be {\emph{preclusive}}, where
\begin{definition}
A co-event $\phi$ is {\emph{preclusive}} if
\begin{equation}
\mu(A)=0 \Rightarrow \phi(A) = 0, \ \ \ \forall A \in \mathfrak{A}\,.
\end{equation}
\end{definition}
\noindent We will also make use of the following definitions:
\begin{definition}
A filter $F\subseteq \mathfrak{A}$ is \emph{preclusive} if none of its elements are null.
\end{definition}
\begin{definition}
An event $A\in \mathfrak{A}$ is \emph{stymied} if it is a subset of a null event.
\end{definition}
The physical world in a classical theory when $\Omega$ is finite is then described equivalently by any of the following.
\begin{enumerate}[(i)]
\item A single history, an element of $\Omega$, which is not an element of any null event.
\item A minimal   non-empty non-stymied event (ordered by inclusion).
\item A preclusive ultrafilter on $\Omega$.
\item A maximal preclusive filter (ordered by inclusion).
\item A preclusive homomorphism $\phi: \mathfrak{A} \rightarrow \mathbb{Z}_2$.
\item A preclusive co-event for which all classical, Boolean rules of inference hold.
 \item A preclusive, unital, MP, C2 co-event.
\item  A minimal preclusive, multiplicative co-event, where minimality is in the order
\begin{equation}\label{eq:order}
\phi_1 \preceq \phi_2  \ \ {\textrm{if}}\ \  \phi_2(A)=1 \Rightarrow \phi_1(A)=1\,.
\end{equation}
\item A minimal preclusive, unital, MP co-event, where again minimality is in the order (\ref{eq:order}).
\end{enumerate}

The equivalence of item (vii) is the import of Theorem \ref{theorem:hom}.
The final two items (viii) and (ix) introduce the concept of \emph{minimality}, which is a finest grainedness condition or a Principle of Maximal Detail: nature affirms as many events as possible without violating preclusion. That the conditions in (viii) imply that $\phi$ is a homomorphism is 
proved by Sorkin \cite{Sorkin10,Sorkin11}.
That (viii) and (ix) are  equivalent is the import of Theorem \ref{theorem:MPfiltermult}.

In a classical theory one is free to consider any or all of these as corresponding to the physical world, since each is equivalent to a single history $\gamma \in \Omega$. 

\subsection{Quantum Measure Theory}\label{sec:reality}
Quantum theories find their place in the framework of GMT at the second level of a countably infinite hierarchy of theories labelled by the amount of interference there is between histories \cite{Sorkin94}.
A quantum measure theory has the three-fold structure described in Section \ref{sec:basics}, just as a classical theory does, and it too is based on a set $\Omega$ of spacetime histories---the histories summed over in the path integral for the theory.
The departure from a classical theory is encoded in the nature of the measure $\mu$ which is in general no longer a probability measure.
Indeed, given by the path integral, a quantal $\mu$ does not satisfy the Kolmogorov sum rule but, rather, a quadratic analogue of it \cite{Sorkin94,Sorkin97,Sorkin07a}.
The existence of interference between histories means that there are quantum measure systems for which the union of all the null events is the whole of $\Omega$.
Examples are the three--slit experiment \cite{Sorkin07a}, the Kochen--Specker antinomy \cite{KochenSpecker67,Bell66}, \cite{Dowker08,Dowker11} and the inequality--free version of Bell's theorem due to Greenberger, Horne and Zeilinger \cite{GHZ89,GHSZ90} and Mermin \cite{Mermin90a,Mermin90b} \cite{MikeDion}.
The condition of preclusion therefore runs into conflict with the proposal that the physical world 
corresponds to a single history since if every history is an element of some null event 
there is no history that can happen: \emph{reductio ad absurdum}.

Choosing to uphold preclusion as a dynamical law means therefore that of the above list of 9 equivalent descriptions of a classical physical world (i)  fails in the quantal setting, and so do (iii), (v), (vi) and (vii).
However, the other 4 --- (ii), (iv), (viii) and (ix) --- survive and remain mutually equivalent for a finite quantal measure theory.
That (ii), (iv) and (viii) are equivalent can be shown using the fact that a multiplicative co-event $\phi$ defines and is defined by its \emph{support}, $F(\phi)\in \mathfrak{A}$, the intersection of all those events that are affirmed by $\phi$:
\begin{equation}
F(\phi) := \bigcap_{S\in \phi^{-1}(1)} S\,.
\end{equation}

Adopting (viii) as the axiom for the possible co-events of a theory gives the resulting scheme its name: the multiplicative scheme. The multiplicative scheme is a unifying proposal: whether classical or quantum, the physical world is a minimal preclusive multiplicative co-event \cite{Sorkin07b}.
What we have shown here is that it could just as well be dubbed the ``modus ponens scheme''.

\section{Final Words}
With hindsight, we can see that the belief that the geometry of physical space was fixed and Euclidean came about because deviations from Euclidean geometry at non-relativistic speeds and small curvatures are difficult to detect.
In a similar vein, Sorkin suggests, the need for deviations from classical rules of inference about physical events lay undetected until the discovery of quantum phenomena (see however \cite{Sorkin10}).
That's all very well, but it could seem much harder to wean ourselves off the structure of classical logic than to give up Euclidean geometry.
To those who feel that classical rules of inference are essential to science this reassurance can be offered: in GMT classical rules of inference are used to reason about the co-events themselves, because a single co-event corresponds to the physical world.

Moreover, in the multiplicative scheme (to the extent that the finite system case is a good guide to the more general, infinite case) each co-event can be equivalently understood in terms of its support---a subset of $\Omega$.
In Hartle's Generalised Quantum Mechanics \cite{Hartle93,Hartle95}, this subset would be called a {\emph{coarse--grained history}}; the proposal of the multiplicative scheme is to describe the physical world as a single coarse--grained history.
The altered rules of inference in the multiplicative scheme for GMT are no more of a conceptual leap than this: the physical world is not as fine grained as it might have been, and there are some details which are missing, ontologically.
\
Furthermore, the results reported here reveal the alteration of logic in the multiplicative scheme to be the mildest possible modification: keeping MP and relinquishing only C2.
Relinquishing C2 in physics means allowing the possibility that an electron is not inside a box and not outside it either. Another example is accepting the possibility that a photon in a double slit experiment does not pass through the left slit and does not pass through the right slit.
At the level of electrons and photons, such a non-classical state of affairs is not too hard to swallow; indeed, very many, very similar statements are commonly made about the microscopic details of quantum systems.
The multiplicative scheme for GMT is a proposal for making precise the nature of Feynman's ``logical tightrope'' and raises the important question: ``Are violations of classical logic confined to the microscopic realm?''.
Answering this question becomes a matter of calculation within any given theory \cite{Sorkin11}.

\begin{acknowledgements}
We thank Rafael Sorkin for helpful discussions. Research at Perimeter Institute for Theoretical Physics is supported in part by the Government of Canada through NSERC and by the Province of Ontario through MRI.
FD and PW are supported in part by COST Action MP1006. PW was supported in part by 
EPSRC grant EP/K022717/1. PW acknowledges support from the University of Athens
during this work.\end{acknowledgements}
\bibliography{global}

\end{document}